\documentclass[reprint,
					aps,
					pra,
					english,
					10pt
]{revtex4-1}

	\usepackage[utf8]{inputenc}
	\usepackage[T1]{fontenc}

	\usepackage[intlimits]{amsmath}

	\usepackage{fixmath}

	\usepackage{amssymb}

	\usepackage{mathtools}
	\mathtoolsset{showonlyrefs}

	\usepackage[caption=false]{subfig}

	\usepackage{dsfont}

	\usepackage{textcomp}

	\usepackage{graphicx}
	
	\usepackage[final,colorlinks=true,linkcolor=blue,citecolor=blue]{hyperref}

	\usepackage{amsthm}

	\theoremstyle{plain}
	\newtheorem{theorem}{Theorem}

	\theoremstyle{definition}
	\newtheorem{definition}{Definition}

  \renewcommand*\d[1]{d #1\,}
  \newcommand{\ii}{i}
  \newcommand*{\per}{\operatorname{perm}}
  \newcommand{\abs}[1]{\left\lvert{#1}\right\rvert}
  \newcommand{\ee}[1]{\operatorname{e}^{#1}}
  \newcommand*{\defeq}{\mathrel{\vcenter{\baselineskip0.5ex \lineskiplimit0pt
                     \hbox{\scriptsize.}\hbox{\scriptsize.}}}%
                     =}
  \renewcommand{\Omega}{\varOmega}
  \renewcommand*\vec[1]{\mathbold{#1}}
  
  \newcommand{\ket}[1]{\vert #1 \rangle}

  \newcommand*{\EV}{\operatorname{E}\displaylimits}

  \makeatletter
  \def\citecheck{%
	  \expandafter\ifx\csname bla\endcsname\relax%
		  \protected@write\@mainaux{}%
		  {\string\checkdef{1}}%
		  \typeout{LaTeX Warning: Label(s) may have changed. Rerun to get cross-references right .}%
	  \else
	  \ifnum\bla=1
		  \protected@write\@mainaux{}%
		  {\string\checkdef{2}}%
		  \typeout{LaTeX Warning: Label(s) may have changed. Rerun to get cross-references right .}%
	  \else
	  \ifnum\bla=2
	  \fi
	  \fi
	  \fi
  }

	\def\checkdef#1{\expandafter\xdef\csname bla\endcsname{#1}}
	\makeatother

\begin{document}
  \title{From the Physics to the Computational Complexity of Multiboson Correlation Interference}
  \author{Simon Laibacher}
  \author{Vincenzo Tamma}
  \email{vincenzo.tamma@uni-ulm.de}
  \affiliation{Institut f\"{u}r Quantenphysik and Center for Integrated Quantum
  Science and Technology (IQ\textsuperscript{ST}), Universit\"at Ulm, D-89069 Ulm, Germany}

  \begin{abstract}  
	  We demonstrate how the physics of \emph{multiboson correlation interference}
	  leads to the computational complexity of linear optical interferometers
	  based on correlation measurements in the degrees of freedom of the input bosons.
	  In particular, we address the task of \emph{MultiBoson Correlation Sampling}
	  (MBCS) from the
	  probability distribution associated with polarization- and time-resolved
	  detections at the output of random linear optical networks. We show that the MBCS problem is fundamentally hard to solve classically
	  even for nonidentical input photons, regardless of the color of the photons, making it also very appealing from an experimental point of view. These results fully manifest 
	  the quantum computational supremacy inherent to the fundamental nature
	  of quantum interference.
  \end{abstract}
  \maketitle

  \citecheck 
\textbf{\textit{Motivation.}}
The interference of multiple bosons based on high-order correlation measurements \cite{TammaLaibacherPRL,TammaLaibacher2,Tamma2015a} in a linear network is a phenomenon that is fundamental in atomic, molecular, and optical physics. The richness of its features 
gives rise to a wide variety of applications in quantum information processing \cite{TammaLaibacherPRL,Pan2012,Knill2001}, quantum metrology \cite{HanburyBrown1956,Motes2015,DAngelo2008}, and imaging \cite{Pittman1995}. Already correlated detections of two bosons after the interaction with a balanced beam splitter reveal an interference effect of truly quantum mechanical origin \cite{Hong1987,Alley1986,*Shih1988,Kaufman2014,Lopes2015}: both particles always end up in the same output port due to the destructive interference of the two-boson quantum paths in which the bosons are either both reflected or both transmitted.

Going to higher-order correlation measurements in optical networks of large dimensions, multiboson interference becomes increasingly complex, promising a computational power that is not achievable classically \cite{aaronson2011computational,Tamma2015}.
Multiphoton correlation experiments with more than two photons have already been performed \cite{Yao2012,Ra2013,Metcalf2013,Broome2013,Crespi2013,tillmann2013experimental,Tillmann2014,Spring2013,Spagnolo2014,Carolan2013,Bentivegna2015},
providing an important milestone towards experiments of higher orders \cite{Franson2013,Ralph2013}.

These experiments are usually based on joint measurements at the interferometer output ports
``classically'' averaging over the photons' degrees of freedom (e.g. time, polarization).
In this context, Aaronson and Arkhipov argued the computational hardness of multiboson interference in linear optics for identical bosons by introducing the well-known boson sampling problem \cite{aaronson2011computational}. Does this computational hardness also occur for nonidentical photons? 
While the computational complexity for partially distinguishable photons is still not known \cite{Tamma2015}, it is clear that boson sampling becomes computationally trivial for fully distinguishable photons when the information about the detection times and polarizations is completely ignored.

However, recent technological advances have enabled experimentalists to produce arbitrarily polarized single photons with near arbitrary spectral and temporal properties \cite{Keller2004,Kolchin2008,Polycarpou2012} which can be ``read out'' by time- and polarization-resolving measurements \cite{TammaLaibacherPRL,Legero2004,Zhao2014,PanYuan2007,PanBao2012} with extremely fast detectors \cite{Eisaman2011}.
This makes it possible to encode entire ``quantum alphabets'' in the degrees of freedom of multiple photons \cite{Nisbet-Jones2013,Monroe2012} and to retrieve the encoded information by correlation measurements in those degrees of freedom, representing a valuable tool in quantum information processing \cite{TammaLaibacherPRL,Tamma2014,Tamm2015,Tamm2015a,tamma2011factoring,Tamma2009,Tamma2012,Bennett1992,Moehring2007,Yuan2008,Greenberger1989,Bouwmeester1999,Bennett1993,Zukowski1993,Holleczek2015}.

All these remarkable technological achievements now allow experimentalists to fully address the following fundamental questions about the interplay between the physics and the complexity of multiboson interference: 
How do the spectral distributions of $N$ nonidentical photons determine the occurrence of $N$-photon interference events in time- and polarization-resolving correlated measurements? How and to what degree is this occurrence connected with computational complexity?
Does computational hardness really disappear for input bosons that are completely distinguishable in their spectra?
This letter aims to answer all these important questions, from both a fundamental and an experimental point of view, demonstrating the inherent computational complexity of the physics of multiboson correlation interference even for nonidentical photons.

\textbf{\textit{MultiBoson Correlation Sampling (MBCS)}}.

We consider $N$ single photons prepared at the $N$ input ports of a linear interferometer (see Fig.~\ref{fig:InterferometerSetup}) with $2M\gg 2N$ ports. The interferometer unitary transformation $\mathcal{U}$ is chosen randomly according to the Haar measure and is implemented by using a polynomial number (in $M$) of passive linear optical elements~\cite{Reck1994}. The state of $N$ single photons injected in a set  $\mathcal{S}$ of $N$ input ports $s \in \mathcal{S}$ is given by
\begin{align}
	\ket{\mathcal{S}} \defeq \bigotimes_{s\in \mathcal{S}}
	\ket{1[\vec{\xi}_{s}]}_{s}
	\bigotimes_{s \notin \mathcal{S}}
	\ket{0}_{s},
	\label{eqn:StateDefinition}
\end{align}
with the single photon states
\begin{align}
	\ket{1[\vec{\xi}_{s}]}_{s} \defeq \sum_{\lambda=1,2} \int_{0}^{\infty} \d{\omega} \left( \vec{e}_{\lambda}\cdot \vec{\xi}_{s}(\omega) \right) \hat{a}^{\dagger}_{s,\lambda} (\omega) \ket{0}_{s},
	\label{eqn:SinglePhotonState}
	\noeqref{eqn:SinglePhotonState}
\end{align}
where $\{\vec{e}_1,\vec{e}_2\}$ is an arbitrary polarization basis and $\hat{a}^{\dagger}_{s,\lambda}(\omega)$ is the creation operator for the frequency mode $\omega$ and the polarization $\lambda$ \cite{Loudon2000}. The complex spectral amplitude 
\begin{align}
	\vec{\xi}_s (\omega) \defeq \vec{v}_s\; \xi_s(\omega-\omega_s) \ee{\ii \omega t_{0s}}
	\label{eqn:ComplexSpectrum}
\end{align}
is defined by the spectral shape $\xi_s(\omega-\omega_s) \in \mathds{R}$ (centered around the central frequency (photon color) $\omega_s$ and with normalization $\int \d{\omega} \abs{\xi_s(\omega)}^2 = 1$), the polarization $\vec{v}_s$, and the time $t_{0s}$ of emission of the photon injected in the port $s\in \mathcal{S}$.
For simplicity, we consider input-photon spectra satisfying the narrow bandwidth approximation and a polarization-independent interferometric evolution with equal propagation time $\Delta t$ for each possible path from an input source to a detector at the interferometer output.

Given such a multiboson interferometer and assuming identical photons, $\vec{\xi}_s = \vec{\xi} \ \forall s\in\mathcal{S}$, the boson sampling problem \cite{aaronson2011computational} was defined by Aaronson and Arkhipov as the task of sampling from the probability distribution over the output port samples $\mathcal{D}$, regardless of detection times and polarizations.
We address here an interesting generalization of this famous problem by introducing the problem of MultiBoson Correlation Sampling (MBCS) \cite{TammaLaibacherPRL,Tamma2015b}.
The MBCS problem is defined as the task of sampling
at the interferometer output from the probability distribution associated with time- and polarization-resolving correlation measurements.
Each possible sample corresponds to an $N$-photon detection event at an $N$-port subset $\mathcal{D}$ of the $M$ output ports at given times and polarizations $\{t_d, \vec{p_d}\}_{d \in \mathcal{D}}$, with $\vec{p}_d \in \{ \vec{e}_1, \vec{e}_2\}$ \footnote{The case of boson bunching at the detectors can be neglected for $M\gg N$ \cite{aaronson2011computational}.}.
\begin{figure} 
  \begin{center}
	  \includegraphics[scale=0.95]{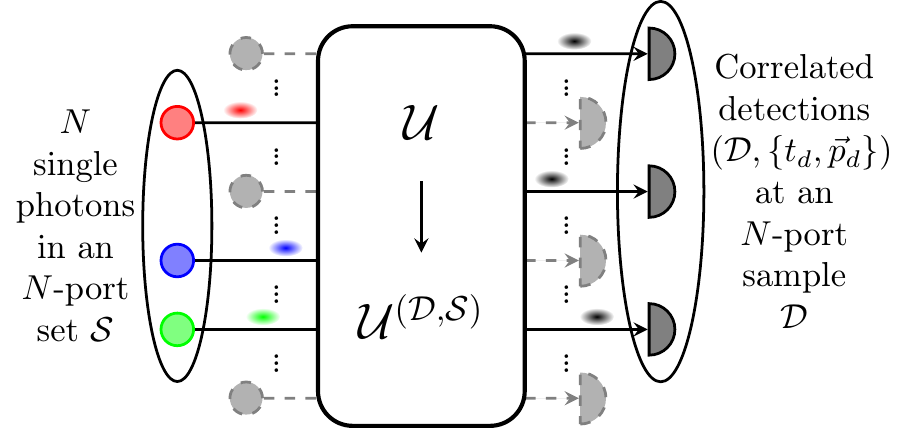}
	  \caption{General setup for multiboson correlation sampling. $N$ single photons are injected into
		  an $N$-port subset $\mathcal{S}$ of the $M\gg N$ input ports of a
		  random linear interferometer. At the output of the interferometer, they are detected in one of the possible port samples $\mathcal{D}$ containing $N$ of the $M$ output ports at corresponding detection times and polarizations
		  $\{t_d,\vec{p}_d\}_{d\in \mathcal{D}}$. For each output port sample $\mathcal{D}$ and given input configuration
		  $\mathcal{S}$, the evolution through the interferometer
		  is fully described by a $N\times N$ submatrix $\mathcal{U}^{(\mathcal{D},\mathcal{S})}$ of the
		  $M \times
		  M$ interferometer matrix $\mathcal{U}$. 
	  }
	  \label{fig:InterferometerSetup}
  \end{center}
\end{figure}
The $N$-photon detection probability rate corresponding to a sample  
$(\mathcal{D},\{t_d, \vec{p_d}\}_{d \in \mathcal{D}})$ depends \cite{TammaLaibacherPRL} on both the
$N \times N$ submatrix
\begin{align}
	\mathcal{U}^{(\mathcal{D},\mathcal{S})} \defeq [ \mathcal{U}_{d,s} ]_{\substack{d\in \mathcal{D} \\ s\in \mathcal{S}}}
\end{align}
of the $M\times M$ unitary matrix $\mathcal{U}$ describing the interferometer, and the Fourier transforms 
\begin{align}
	\vec{\chi}_s(t) &\defeq \mathcal{F}[\vec{\xi}_s](t-\Delta t) \\ 
	&= \vec{v}_s \; \chi_s(t-t_{0s}-\Delta t) \ee{\ii \omega_{s} (t-t_{0s} -\Delta t)} 
	\label{eqn:Fourier}
\end{align}
of the single-photon spectra $\vec{\xi}_s(\omega)$ in Eq.~\eqref{eqn:ComplexSpectrum} (with $\chi_s(t)$ being the Fourier transform of $\xi_s(\omega)$). Defining 
the matrices
\begin{align}
	\mathcal{T}^{(\mathcal{D},\mathcal{S})}_{\{t_d,\vec{p}_d\}} \defeq
	\big[ \mathcal{U}_{d,s} \;\big( \vec{p}_d \cdot \vec{\chi}_s(t_d) \big)
	\big]_{\substack{d\in \mathcal{D} \\ s\in \mathcal{S}}}
		\label{eqn:CorrMatrixDefintion}
\end{align}
and using the definition
\begin{align}
\per \mathcal{M} \defeq \sum_{\sigma \in \Sigma_N} \prod_{i=1}^N \mathcal{M}_{i,\sigma(i)}
\label{eqn:PermDef} 
\end{align}
of the permanent of a matrix $\mathcal{M}$,
where the sum runs over all permutations $\sigma$ in the symmetric group $\Sigma_N$, the probability rate of an $N$-fold detection event $(\mathcal{D},\{t_d, \vec{p_d}\}_{d \in \mathcal{D}})$ is 
\begin{align}
	G^{(\mathcal{D},\mathcal{S})}_{\{t_d, \vec{p}_d\}}
	&= \abs{\per \mathcal{T}^{(\mathcal{D},\mathcal{S})}_{\{t_d,\vec{p}_d\}}}^2 
	,
	\label{eqn:CorrelationFinal}
\end{align}
for ideal photodetectors.

By considering an integration time $T_I$ short enough such that
 \begin{multline}
	 \forall t_d: \chi_s (t- t_{0s} -\Delta t) \chi_{s'} (t - t_{0s'} - \Delta t) \ee{\ii (\omega_{s} - \omega_{s'}) t} \approx \text{const.}  \\
	 \quad \forall t \in [t_d-T_I,t_d+T_I], \forall s,s' \in \mathcal{S},
	 \label{eqn:ConditionShortIntegrationTime}
 \end{multline}
we obtain, for a detection sample $(\mathcal{D},\{t_d, \vec{p}_d\}_{d \in \mathcal{D}})$, the probability 
\begin{align}
	P^{(\mathcal{D},\mathcal{S})}_{\{t_d, \vec{p}_d\}}
	&= (2 \,T_I)^N \big|\!\per \mathcal{T}^{(\mathcal{D},\mathcal{S})}_{\{t_d,\vec{p}_d\}}\big|^2 
	\label{eqn:ProbFinal}
\end{align}
of an $N$-fold detection in the time intervals $\{[t_d-T_I,t_d+T_I]\}_{d \in \mathcal{D}}$, where the detection time axes are discretized with step width $2 T_I$.

We emphasize that, for each possible sample $(\mathcal{D},\{t_d,\vec{p}_d\}_{d\in\mathcal{D}})$, the probability in Eq.~\eqref{eqn:ProbFinal} is at most exponentially small in $N$, as demonstrated in Theorem~1 in the Supplemental Material \footnote{The Supplemental Material can be found at the end of this file and contains the references \cite{Stockmeyer1983,Toda1991,Kuperberg2015}}.

\textbf{\textit{Exact MBCS}}.
Obviously, the complexity of sampling exactly from the probability distribution defined by Eq.~\eqref{eqn:ProbFinal} depends on the $N$-tuples $\{\vec{\xi}_s\}_{s\in\mathcal{S}}$ of single-photon input spectra in Eq.~\eqref{eqn:ComplexSpectrum} \footnote{For a more mathematical definition of the MBCS problem in the exact and the approximate case, we refer to the Supplemental Material \cite{Note2}}.

With that in mind, in order to establish the complexity of exact MBCS,
it is useful to define the $N$-photon interference matrix with elements
\begin{align}
	a(s,s') &\defeq \abs{\vec{v}_s \cdot \vec{v}_{s'}} \int_{-\infty}^{\infty} \d{t} \abs{\chi_s(t-t_{0s})} \abs{\chi_{s'}(t-t_{0s'})} \leq 1,
	\label{eqn:TwoPhotonOverlapModulus}
\end{align}
with $s, s' \in \mathcal{S}$, depending on the pairwise overlaps of the absolute values of the temporal single-photon detection amplitudes \cite{Glauber2006} $\chi_s(t-t_{0s}-\Delta t) \ee{\ii \omega_s (t-t_{0s}-\Delta t)}$ and of the polarizations $\vec{v}_s$ in Eq.~\eqref{eqn:Fourier} .
For non-vanishing elements
\begin{align}
	0 < a(s,s') \leq 1 \quad \forall s,s' \in \mathcal{S},
	\label{eqn:NonvanishingOverlap}
\end{align}
there exists a time interval $T$ and at least a polarization $\vec{e}_{\bar{\lambda}} \in \{ \vec{e}_1,\vec{e}_2 \}$, such that 
$$\vec{e}_{\bar{\lambda}}\cdot \vec{\chi}_s(t_d) \neq 0 \,\,\,\,\,\forall t_d\in T, \forall s \in \mathcal{S}, \forall d \in \mathcal{D}.$$ 

It is then ensured that for each detection sample $(\mathcal{D},\{t_d, \vec{p_d}\}_{d \in \mathcal{D}})$, with $t_d \in T, \vec{p}_d = \vec{e}_{\bar{\lambda}} \ \forall d\in \mathcal{D}$, the input photons are indistinguishable at the detectors: 
this leads to the interference of all possible $N!$ $N$-photon quantum paths manifested by the coherent superposition of all corresponding, non-vanishing $N!$ $N$-photon detection amplitudes in Eq.~\eqref{eqn:CorrelationFinal}.
Therefore, only the conditions~\eqref{eqn:ConditionShortIntegrationTime} and \eqref{eqn:NonvanishingOverlap} for the nonidentical input spectra $\{\vec{\xi}_s\}_{s\in\mathcal{S}}$ in Eq.~\eqref{eqn:ComplexSpectrum} are enough to ensure the occurrence of $N$-photon correlation interference events.

Even more interestingly, the same simple conditions lead to the computational hardness of the exact MBCS problem, establishing a connection between the occurrence of multiphoton correlation interference and complexity. Indeed, for approximately equal detection times $t_d\approx t \in T$ and equal polarizations $\vec{p}_d=\vec{e}_{\bar{\lambda}},\, \forall d\in \mathcal{D}$,
the multiphoton detection probabilities in Eq.~\eqref{eqn:ProbFinal} become
\begin{align}
	P^{(\mathcal{D},\mathcal{S})}_{\{t_d, \vec{p}_d\}}
	= \abs{\per \mathcal{U}^{(\mathcal{D},\mathcal{S})}}^2 (2 T_I)^N \prod_{s\in \mathcal{S}}  \abs{ \vec{e}_{\bar{\lambda}}\cdot \vec{\chi}_s(t)}^2.
	\label{eqn:RateIdenticalDets}
\end{align}
The interference of all $N$-photon quantum paths in Eq.~\eqref{eqn:RateIdenticalDets} depends, apart from an overall factor, only on the permanent of a submatrix $\mathcal{U}^{(\mathcal{D},\mathcal{S})}$ of the interferometer random unitary matrix $\mathcal{U}$. 
For $N\ll M$, these matrices have elements given by approximately independent and identically distributed (i.i.d.) Gaussian random variables and the approximation of their respective permanents is a \#P-hard task \cite{aaronson2011computational}. 
We emphasize that the presence of only an arbitrarily small fraction of samples with probabilities as in Eq.~\eqref{eqn:RateIdenticalDets} would be enough to ensure the hardness of the exact MBCS. This can be shown analogously to the hardness proof of the original problem of exact boson sampling in \cite{aaronson2011computational}.
Indeed, the ability to perform exact MBCS with a polynomial number of resources would imply that the task of approximating any given, fixed permanent associated with the probability distribution~\eqref{eqn:ProbFinal} is in the complexity class $\mathrm{BPP^{NP}}$. Since this would also include the task of approximating the \#P-hard permanents emerging in Eq.~\eqref{eqn:RateIdenticalDets}, the polynomial hierarchy would collapse to the third level, which is strongly believed to be highly unlikely.
We refer to section II of the Supplemental Material for more details \cite{Note2}.

Interestingly, differently from the original boson sampling problem \cite{aaronson2011computational}, the classical intractability of exact MBCS is not conditioned on input photons with approximately identical spectra $\vec{\xi}_s$ in Eq.~\eqref{eqn:ComplexSpectrum}. Only the simple conditions~\eqref{eqn:ConditionShortIntegrationTime} and \eqref{eqn:NonvanishingOverlap} on the spectra are enough to guarantee its computational hardness. 

 \textbf{\textit{Approximate MBCS.}}
 Is approximate MBCS also not tractable with a classical computer? Such a question is obviously of fundamental importance from an experimental point of view, since it takes into account the inevitable experimental errors in an MBCS quantum interferometer which make only approximate sampling possible \cite{Note3}.
 We consider, for simplicity, the case of an $N$-photon interference matrix in Eq.~\eqref{eqn:TwoPhotonOverlapModulus} with unit elements
\begin{align}
	 a(s,s') \cong 1 \quad \forall s,s' \in \mathcal{S}.
	 \label{eqn:FullTemporalOverlap}
\end{align}
This corresponds to two possible scenarios. Either all the input photons are completely identical or they differ only by their color, i.e. central frequency. In these cases the input photons have equal polarizations and are always indistinguishable at the detectors independently of the detection times and polarizations.

To simplify the expressions, we consider here polarization-insensitive-detectors.

\paragraph{Identical input photons.}
For approximately identical frequency spectra
\begin{align}
	\vec{\xi}_s(\omega) &\cong \vec{\xi}(\omega) \quad \forall s \in \mathcal{S},
\end{align}
by using Eq.~\eqref{eqn:ProbFinal}, the polarization-insensitive detection probability reads
\begin{align}
	P^{(\mathcal{D},\mathcal{S})}_{\{t_d\}} &\defeq \sum_{\{\vec{p}_d\}\in \{\vec{e}_1,\vec{e}_2\}^N} P^{(\mathcal{D},\mathcal{S})}_{\{t_d,\vec{p}_d\}} \\
	&= \abs{\per \mathcal{U}^{(\mathcal{D},\mathcal{S})}}^2 (2T_I)^{N} \prod_{d\in \mathcal{D}} \abs{\vec{\chi}(t_d)}^2,
\end{align}
where we used the property ${\sum_{\vec{p}_d=\vec{e}_1,\vec{e}_2}\abs{\vec{p}_d\cdot \vec{v}} = \abs{\vec{v}}^2 = 1}$.
Of course the only possible events occur within a detection-time interval where the function 
$\abs{\vec{\chi}(t_d)}= \abs{\mathcal{F}[\vec{\xi}](t_d-\Delta t)}$ is not negligible. Here, independently of the detection times $\{t_d\}_{d\in \mathcal{D}}$, all the probability rates associated with each possible sample $(\mathcal{D},\{t_d\}_{d\in \mathcal{D}})$ are given, apart from a prefactor, by the permanents of $N \times N$ submatrices $\mathcal{U}^{(\mathcal{D},\mathcal{S})}$ of the interferometer transformation $\mathcal{U}$. 

When the observer ignores the information about the detection times the approximate MBCS problem reduces to the well known standard formulation of the approximate boson sampling problem, which Aaronson and Arkhipov argued to be intractable with a classical computer \cite{aaronson2011computational}.  
Therefore, the approximate MBCS problem  is at least as complex as the original approximate boson sampling problem.

\paragraph{Photons of different colors.}

We now address the case of input photons in Eq.~\eqref{eqn:SinglePhotonState} with spectral distributions
\begin{align}
	\vec{\xi}_s(\omega) = \vec{v} \; \xi(\omega-\omega_s) \ee{\ii \omega t_0},
\end{align}
with equal emission times $t_{0s} = t_0$ and equal polarizations $\vec{v}_s = \vec{v}$ but different colors $\omega_s$. For simplicity, we consider spectral shapes
\begin{align}
	\xi(\omega) = \frac{1}{\sqrt{\pi \Delta\omega}} \operatorname{sinc} \left( \frac{\omega}{\Delta\omega} \right)
	\label{}
\end{align}
with equal bandwidths $\Delta\omega_{s} = \Delta\omega \lesssim \abs{\omega_s - \omega_{s'}} \ \forall s,s'$, where $\operatorname{sinc} x \defeq \sin x / x$. The $N$-photon interference at the detectors is therefore characterized by the Fourier transforms
\begin{align}
	\vec{\chi}_s(t) = 
	\vec{v} \sqrt{\frac{\Delta\omega}{2}} \operatorname{rect} \left( \frac{\Delta\omega(t-t_{0}-\Delta t)}{2} \right) \ee{\ii \omega_{s} (t-t_{0}-\Delta t)},
	\label{}
\end{align}
with the rectangular function
\begin{align}
	\operatorname{rect} x \defeq
	\begin{dcases}
		1 & \abs{x} < \frac{1}{2} \\
		\frac{1}{2} & \abs{x}  = \frac{1}{2} \\
		0 & \text{else}
	\end{dcases} \quad .
\end{align}
Therefore, the condition~\eqref{eqn:FullTemporalOverlap} is satisfied, and the probability rates in Eq.~\eqref{eqn:CorrelationFinal} are non-vanishing only for detection times
$t_d \in T \defeq [t_0 + \Delta t - 1/\Delta\omega, t_0 + \Delta t + 1/\Delta\omega]\ \forall d\in \mathcal{D}$.
Moreover, Eq.~\eqref{eqn:ConditionShortIntegrationTime} is fulfilled for integration times 
\begin{align}
T_I \ll \abs{\omega_s - \omega_{s'}}^{-1} \ \forall s,s' \in \mathcal{S},
\label{eqn:ShortIntegrationTime}
\end{align}
where $(T_I \Delta\omega)^{-1}$ defines the number of discrete steps of length $2T_I$ along the time interval $T$.
As is known, detectors with such high time-resolution cannot distinguish photons of different colors $\omega_s$ and multiphoton interference can be observed.
Indeed, from Eq.~\eqref{eqn:ProbFinal}, the polarization-independent detection probabilities are
\begin{align}
	P^{(\mathcal{D},\mathcal{S})}_{\{t_d\}} = 
	\left( \Delta\omega\,T_I \right)^{N} \abs{\per \Big( \big[ \mathcal{U}^{(\mathcal{D},\mathcal{S})}_{d,s} \ee{\ii \omega_s t_d} \big]_{\substack{d\in \mathcal{D}\\ s\in \mathcal{S}}} \Big) }^2
	\label{eqn:ProbDiffCol}
\end{align}
for all possible detection time intervals $[t_d - T_I,t_d + T_I]\subset T$. Such probabilities are proportional to permanents of matrices whose elements are the elements of $\mathcal{U}^{(\mathcal{D},\mathcal{S})}$ multiplied by the complex phases $\exp(\ii \omega_s t_d)$.

Since the elements $\mathcal{U}^{(\mathcal{D},\mathcal{S})}_{d,s}$ of the submatrices $\mathcal{U}^{(\mathcal{D},\mathcal{S})}$ are i.i.d. Gaussian random variables and the phase factors $\ee{\ii \omega_s t_d}$ only rotate such elements in the complex plane, the entries of the matrices $[ \mathcal{U}^{(\mathcal{D},\mathcal{S})}_{d,s} \ee{\ii \omega_s t_d} ]_{\substack{d\in \mathcal{D} \\ s\in \mathcal{S}}}$ are also i.i.d. Gaussian random variables as shown in App.~C of the Supplemental Material \cite{Note2}.
Therefore, the probability distribution of the interferometer output interestingly depends, for all possible samples, on permanents whose approximation to within a multiplicative factor is a \#P-hard problem \cite{aaronson2011computational}.
Consequently, even for input photons of different colors, it is possible to show in analogy with Ref.~\cite{aaronson2011computational} that approximate MBCS is of at least the same complexity as the standard boson sampling with identical photons
\footnote{
We refer to section III of the Supplemental Material \cite{Note2} for a more detailed formal proof.
}.
As a ``bonus'', the number of possible samples $(\mathcal{D},\{t_d\}_{d\in \mathcal{D}})$ is exponentially larger (by a factor $(T_I \Delta\omega)^{-N}$ with $T_I \Delta\omega \ll 1$ according to Eq.~\eqref{eqn:ShortIntegrationTime}) with respect to the standard boson sampling problem.

Does approximate MBCS retain its complexity even for photons which are completely pairwise distinguishable in their colors $\omega_s$ (i.e. $\omega_s - \omega_{s'} \gg \Delta\omega \ \forall s \neq s'$)? We first emphasize that, since these photons are characterized by a pairwise overlap
\begin{align}
	\int_{0}^{\infty} \d{\omega} \vec{\xi}_s(\omega) \cdot \vec{\xi}_{s'}(\omega) \cong 0 \,\,\, \ \forall s \neq s',
	\label{eqn:OverlapComplex}
\end{align}
the approximate boson sampling problem is trivial \cite{aaronson2011computational}. Indeed, in this case, by averaging the rate in Eq.~\eqref{eqn:CorrelationFinal} over all possible detection times and polarizations, one finds that
the boson sampling probability \cite{TammaLaibacherPRL}
\begin{align}
	P^{(\mathcal{D},\mathcal{S})} 
	&= \per \big[ \big| \mathcal{U}^{(\mathcal{D},\mathcal{S})}_{d,s}\big|^2 \big]_{\substack{d\in \mathcal{D} \\ s \in \mathcal{S}}},
	\label{}
\end{align}
for an output port sample $\mathcal{D}$, is given by
the permanent of a non-negative matrix that can be approximated with a polynomial number of resources \cite{Jerrum2004}.
Consequently, one might guess that also the approximate MBCS is computationally trivial.
Nonetheless, the complexity emerging from the result in Eq. \eqref{eqn:ProbDiffCol} is independent of the colors $\omega_s$ of the input photons, demonstrating that also in this case approximate MBCS is classically intractable.

Two essential physical aspects are behind the demonstrated complexity of approximate MBCS:
\emph{all} possible detection-time events can be an outcome of the sampling experiment (none of the events is disregarded) and \emph{all} these time samples arise from the interference of $N!$ multiphoton quantum paths. In conclusion, the physics of sampling among all possible $N$-photon interference events behind our proposal is at the heart of the complexity of approximate MBCS.

\textbf{\textit{Discussion}}. 
In this letter, we demonstrated how and to what degree the occurrence of multiphoton interference in time- and polarization-resolving correlation measurements leads to computational hardness in linear optical interferometers.

The definition of an $N$-photon interference matrix $a(s,s')$ in Eq.~\eqref{eqn:TwoPhotonOverlapModulus} allowed us to formulate the simple sufficient condition~\eqref{eqn:NonvanishingOverlap} on the spectra of the input photons for the occurrence of $N$-photon interference, provided sufficiently small integration times (see Eq.~\eqref{eqn:ConditionShortIntegrationTime}).

Remarkably, these two simple conditions are also sufficient to guarantee the complexity of exact MBCS. In contrast, the complexity of the original exact boson sampling problem has only been proven for identical input photons.

For approximate MBCS on the other hand, not only the existence of samples exhibiting full $N$-photon interference (guaranteed by~\eqref{eqn:NonvanishingOverlap}) is important but also their fraction with respect to the total number of samples.
Interestingly, this is encoded in the magnitude of the entries $a(s,s')$ of the $N$-photon interference matrix in Eq.~\eqref{eqn:TwoPhotonOverlapModulus}.

It was thus natural to consider the simple case of full overlap of the modulus of the single-photon detection amplitudes ($a(s,s') = 1$) where all possible detection events correspond to $N$-photon interference samples
\footnote{The case $0 < a(s,s') < 1$, where the number of $N$-photon interference events is only a finite fraction of the total number of possible events, is beyond the scope of this letter and will be addressed in a future publication.}.
In this case, corresponding to identical input photons or photons with arbitrary colors, approximate MBCS is at least of the same complexity as boson sampling with identical photons.

This is particularly interesting if the differences in the central frequencies are much larger than the width of the single photons' spectral shapes, corresponding to fully distinguishable photons in the sense of Eq.~\eqref{eqn:OverlapComplex}. While approximate boson sampling becomes trivial in this case \cite{aaronson2011computational}, approximate MBCS is at least as complex as when perfectly identical input photons are used.

Since detectors with high temporal resolution ($<1\,$ns) and single photons with large coherence times ($> 1$\,\textmu s) are readily available today experimentally \cite{Eisaman2011}, the requirement of time-resolved measurements in the implementation of the MBCS problem can be readily fulfilled.
Moreover, an implementation of MBCS has the advantage to ease the difficulties faced in the production of identical photons. 
Indeed, photons of approximately equal colors ($\Delta\omega \gg \abs{\omega_s-\omega_{s'}} \ \forall s\neq s'$) are not needed any more, unlike in the original approximate boson sampling problem. This furthermore paves the way towards the use of photons of arbitrarily small bandwidth $\Delta\omega$, where the indistinguishability in the emission times ($1/\Delta\omega \gg \abs{t_{0s}-t_{0s'}} \ \forall s,s'$) can be easily achieved.

In conclusion, all these results represent an important stepping-stone towards a full fundamental understanding of the complexity of multiphoton interference of photons of arbitrary spectra in linear optical networks, when the information about detection times and polarizations is not ignored. This may lead to ``real world'' applications in quantum information processing \cite{Holleczek2015} and in quantum optics overcoming the experimental challenge in the production of identical bosons.

Finally, our results can be extended to bosonic interferometric networks with atoms
\cite{Kaufman2014,Lopes2015,Shen2014}, plasmons \cite{Varro2013} or mesoscopic many-body systems \cite{Urbina2014} and are also relevant to the study of the complexity of multiboson correlation interference for different input states \cite{TammaLaibacher2,Lund2014,Olson2015} and different correlation measurements \cite{Glauber2010}.

\begin{acknowledgments}
The authors are very grateful to S. Aaronson for useful insights and discussions, as well as to K. Ranade for providing insights on the theory of i.i.d. Gaussian matrices.

V.T. acknowledges the support of the German Space Agency DLR with funds
provided by the Federal Ministry of Economics and Technology (BMWi) under
grant no. DLR 50 WM 1556.

This work was supported by a grant from the Ministry of Science, Research and the Arts of Baden-W\"urttemberg (Az: 33-7533-30-10/19/2).

The authors contributed equally to this letter.
\end{acknowledgments}

%

\clearpage
\widetext
\begin{center}
\textbf{\large From the Physics to the Computational Complexity of Multiboson Correlation
  Interference: Supplemental Material}
\end{center}
\setcounter{equation}{0}
\setcounter{figure}{0}
\setcounter{table}{0}
\setcounter{page}{1}
\makeatletter
\renewcommand{\theequation}{S\arabic{equation}}
\renewcommand{\thefigure}{S\arabic{figure}}
\renewcommand{\bibnumfmt}[1]{[S#1]}
\renewcommand{\citenumfont}[1]{S#1}
\pagenumbering{Roman}

  \section{The MultiBoson Correlation Sampling (MBCS) problem}

  Here, we formally define the MultiBoson Correlation Sampling (MBCS) problem in linear interferometers described by a unitary $M\times M$ matrix chosen randomly according to the Haar measure.

  We first introduce the notion of families of interferometer input states $\{|1[\vec{\xi}_s]\rangle_s\}_{s\in\mathcal{S}}$, at the input ports $s\in \mathcal{S}$, defined by the sets $\{\vec{\xi}_s\}_{s\in \mathcal{S}}$ of $N$ complex spectral distributions 
  \begin{align}
	  \vec{\xi}_s(\omega) \defeq \vec{v}_s \xi_{s}(\omega-\omega_s) \ee{\ii \omega t_{0s}},
	  \label{eqn:Supp:ComplexSpectra}
  \end{align}
  with square integrable spectral shapes $\xi_s(\omega-\omega_s)$ centered around $\omega=\omega_s$, polarizations $\vec{v}_s$, central frequencies $\omega_s$, and emission times $t_{0s}$:

  \begin{definition}[Families of input states]
	  Let $I$ be the set of square normalized, complex spectral distributions which fulfill the narrow-bandwidth approximation. For given sets $F_N \subseteq I^N$, with $N \in \mathds{N}$, of $N$-tuples $\{\vec{\xi}_{s}\}_{s\in\mathcal{S}} \in F_N$ of single-photon spectra, we define the family $F \defeq \{F_N\}_{N\in \mathds{N}}$. Further, we denote as $\mathcal{F}$ the set of all possible families $F$.
	  \label{def:Families}
  \end{definition}

  Secondly, we formally define all the possible samples that can be detected at the interferometer output in the MBCS problem by discretizing the detection-time axis into bins of width $2T_I$ (integration time) much smaller than the temporal widths of the photons and centered at times $t_k\defeq  2 T_I k$. Here, $k=k_{\text{min}},k_{\text{min}}+1,\dots,k_{\text{max}}$, where $k_{\text{min}}, k_{\text{max}} \in \mathds{Z}$ define the temporal interval where detections can occur. From now on, we will refer to the time samples $\{t_d\}_{d\in\mathcal{D}}$ only in terms of the $N$ integers $\{k_d = t_{d}/(2T_I)\}_{d\in \mathcal{D}}$ and define each possible \emph{overall sample} as $\mathfrak{S}\defeq(\mathcal{D};\{k_d\},\{\vec{p}_d\})$, with $\vec{p}_d\in\{\vec{e}_1,\vec{e}_2\}$ for a fixed orthonormal polarization basis $\{\vec{e}_1,\vec{e}_2\}$.

  The probability distribution
  \begin{align}
	  \mathcal{P}_F= \mathcal{P}_F(\{\vec{\xi}_s\}\in F_N,A\in \mathfrak{U}_{M,N})
	  \label{eqn:Supp:Distribution}
  \end{align}
  associated with all the possible samples $\mathfrak{S}$ is defined by a given rectangular $M\times N$ submatrix
  \begin{align}
	  A \defeq [\mathcal{U}_{d,s}]_{\substack{d=1,\dots,M \\ s\in \mathcal{S}}} \in \mathfrak{U}_{M,N}
  \end{align}
  of the interferometer matrix $\mathcal{U}$,
  where $\mathfrak{U}_{M,N}$ is the set of $M\times N$ matrices with orthogonal columns,
  and by a given set of complex spectral distributions $\{\vec{\xi}_s\}\in F_N$ for the $N$ input photons, with $F = \{F_N\}_{N\in \mathds{N}} \in \mathcal{F}$.
  
  As shown in the main letter, defining the matrices
  \begin{align}
	  \mathcal{T}^{(\mathcal{D},\mathcal{S})}_{\{k_d,\vec{p}_d\}} = \mathcal{T}^{(\mathcal{D},\mathcal{S})}_{\{t_{k_d},\vec{p}_d\}} \defeq \Big[ \mathcal{U}_{d,s} \big( \vec{p}_d \cdot \vec{\chi}_s(2T_I k_d)\big) \Big]_{\substack{d\in \mathcal{D}\\ s\in \mathcal{S}}},
	  \label{eqn:Supp:Tmatrix}
  \end{align}
	where
  \begin{align}
	  \vec{\chi}_s(t) = \vec{v}_s \chi_s(t-t_{0s}-\Delta t) \ee{\ii \omega_{s}(t- t_{0s} -\Delta t)}
	  \label{eqn:Supp:TemporalSpectra}
  \end{align}
  are the Fourier transforms of the spectra $\vec{\xi}_s(\omega)$ in Eq.~\eqref{eqn:Supp:ComplexSpectra} ($\Delta t$ is the delay the photons pick up in the interferometer),
  the probability $p_{\mathfrak{S}}$ associated with the sample $\mathfrak{S}$, when sampling from the distribution $\mathcal{P}_F$ in Eq.~\eqref{eqn:Supp:Distribution}, is
  \begin{align}
	  p_{\mathfrak{S}} \defeq \Pr_{\mathcal{P}_F}[\mathfrak{S}]= P^{(\mathcal{D},\mathcal{S})}_{\{t_{k_d},\vec{p}_d\}} 
	  = (2T_I)^N \abs{\per \mathcal{T}^{(\mathcal{D},\mathcal{S})}_{\{k_d,\vec{p}_d\}}}^2. 
	  \label{eqn:Supp:ProbabilitiesExact}
  \end{align}
  For samples $\mathfrak{S}_{\text{id}}$ with identical detection times $k_d = k \ \forall d\in\mathcal{D}$ and polarizations $\vec{p}_d = \vec{p} \ \forall d\in\mathcal{D}$, this simplifies to
  \begin{align}
	  p_{\mathfrak{S}_{\text{id}}} = (2T_I)^N \prod_{s \in \mathcal{S}} \abs{\vec{p} \cdot \vec{\chi}_s(2T_I k)}^2 \abs{\per \mathcal{U}^{(\mathcal{D},\mathcal{S})} }^2.
	  \label{eqn:Supp:ProbabilityIdenticaTimePol}
  \end{align}
 
  Following Aaronson and Arkhipov in \cite{S_aaronson2011computational}, it is reasonable to represent the elements of the interferometer matrix as rational numbers $(x+\ii y)/2^{\text{poly}(N)}$ with integers $x$ and $y$. Analogously, the same holds for the values $T_I$, $\chi_{s}(t_d-t_{0s}-\Delta t)$, $\vec{p}_d \cdot \vec{v}_s$, and $\exp(\ii \omega_s (t_d - t_{0s}-\Delta t))$, with $s\in\mathcal{S}$
  (notice that this representation is reasonable since all these given complex values describing an MBCS experiment can be approximated with high enough precision by choosing $\text{poly}(N)$ sufficiently large).

 We emphasize that, while we will prove the hardness of the MBCS problem assuming the rational representation of these numbers, this proof can easily be generalized to a representation in terms of algebraic numbers which are
 dense in $\mathds{C}$ (more information can be found in App.~\ref{app:Supp:Algebraic}).

  \begin{theorem}[Exponential size of the probabilities]
  Assume that the interferometer matrix and the temporal distributions are given in a rational representation, as described above. Then, independently of the specific choice of $A\in \mathfrak{U}_{M,N}$, $F=\{F_N\}\in \mathcal{F}$, and $\{\vec{\xi}_s\}_{s\in\mathcal{S}}\in F_N$, all the probabilities~\eqref{eqn:Supp:ProbabilitiesExact} in the probability distribution $\mathcal{P}_F$ are at most exponentially small,
  \begin{align}
	  p_{\mathfrak{S}} \geq 2^{-\mathrm{poly}(N)} \ \forall \mathfrak{S}.
	  \label{}
  \end{align}

  Alternatively, the same result holds if we instead assume that the values are represented by algebraic numbers (which are dense in $\mathds{C}$).
  \label{thm:Supp:ExpSize}
 \end{theorem}
 
 \begin{proof}
	 This is demonstrated in App.~\ref{app:Supp:Binary}.
 \end{proof}
  
 Furthermore, the probabilities can be represented as (see App.~\ref{app:Supp:Binary})
 \begin{align}
	 p_{\mathfrak{S}} = \frac{w_{\mathfrak{S}}}{2^{\text{poly}(N)}}
	 \label{eqn:Supp:BinaryProbabilities}
 \end{align}
 with an integer $w_{\mathfrak{S}}$, which
 ensures that an MBCS oracle as it will be defined in Definition~\ref{def:Supp:Oracle} (using a random input string $r$ of length at most polynomial in $N$) is able to sample from a probability distribution equal or arbitrarily close to the exact probability distribution in Eq.~\eqref{eqn:Supp:Distribution} where the probability for each sample $\mathfrak{S}$ is defined by Eq.~\eqref{eqn:Supp:ProbabilitiesExact}.

  We can finally define an MBCS oracle as:
 \begin{definition}[Definition of an MBCS oracle in analogy to \cite{S_aaronson2011computational}]
	 Let $\mathcal{O}_{F}$ be an oracle that takes as input an $M\times N$ matrix $A\in \mathfrak{U}_{M,N}$, an $N$-tuple $\{\vec{\xi}_s\}_{s\in\mathcal{S}}\in F_N$ from a given family $F\in \mathcal{F}$, an error bound $\beta>0$ (encoded in unary as $0^{1/\beta}$ to ensure a scaling of resources with $1/\beta)$), and a string $r\in \{0,1\}^{\text{poly}(N)}$ which is its only source of randomness. 
	 Let $\mathcal{P}_{\mathcal{O}_F} = \mathcal{P}_{\mathcal{O}_F}(\{\vec{\xi}_s\},A;\beta)$ be the distribution over the outputs of $\mathcal{O}_F$ if $A$, $\{\vec{\xi}_s\}$, and $\beta$ are fixed but $r$ is uniformly random. Then, $\mathcal{O}_F$ is called an \emph{exact MBCS oracle for the family $F$} if $\mathcal{P}_{\mathcal{O}_F}(\{\vec{\xi}_s\},A;\beta)  \equiv \mathcal{P}_F(\{\vec{\xi}_s\},A)$ for all $N\in \mathds{N}$, $A\in \mathfrak{U}_{M,N}$, $\{\vec{\xi}_s\} \in F_N$, and $\beta>0$. Further, $\mathcal{O}_F$ is an \emph{approximate MBCS oracle for the family $F$} if $\|\mathcal{P}_{\mathcal{O}_F}(\{\vec{\xi}_s\},A;\beta) - \mathcal{P}_F(\{\vec{\xi}_s\},A)\| \leq \beta$ for all $N\in \mathds{N}$, $A\in \mathfrak{U}_{M,N}$, $\{\vec{\xi}_s\}\in F_N$, and $\beta>0$, where $\|\mathcal{P}_{\mathcal{O}_F}(\{\vec{\xi}_s\},A;\beta)-\mathcal{P}_F(\{\vec{\xi}_s\},A)\| \defeq 1/2 \sum_{\mathfrak{S}}\big|\Pr_{\mathcal{P}_{\mathcal{O}_F}}[\mathfrak{S}] - \Pr_{\mathcal{P}_F}[\mathfrak{S}]\big|$.
	  \label{def:Supp:Oracle}
  \end{definition}

  \section{The complexity of the exact MBCS problem}\label{sec:Supp:ExactCase}

  Here, we formally demonstrate the complexity of the exact MBCS problem for the set $\mathcal{E}$ of families $E$ defined as:

  \begin{definition}[Set $\mathcal{E}$ of ``complex'' families in the exact MBCS problem]
	  We define $\mathcal{E} \subseteq \mathcal{F}$ (with $\mathcal{F}$ defined in Definition~\ref{def:Families}) as the subset of families $E \defeq \{E_N\}_{N\in \mathds{N}}$ which fulfill the condition
	  \begin{align}
		  \forall N \in \mathds{N}, \forall \{\vec{\xi}_s\}\in E_N: \qquad 0 < a(s,s') \leq 1 \quad \forall s,s'\in \mathcal{S},
		  \label{eqn:Supp:ConditionExactComplexityContinuous}
	  \end{align}
	  with $a(s,s') \defeq \abs{\vec{v}_s \cdot \vec{v}_{s'}} \int_{-\infty}^{\infty} \d{t} \abs{\chi_{s}(t - t_{0s})} \abs{\chi_{s'}(t-t_{0s'})}$.
	  \label{def:Supp:ComplexFamiliesExact}
  \end{definition}
  We first introduce the following theorem:
  \begin{theorem}
	  Let $E = \{E_N\}_{N\in \mathds{N}} \in \mathcal{E}$.
  It is then ensured that a sequence $\big\{(k_N,\vec{p}_N)\big\}_{N\in \mathds{N}}$ of time-polarization tuples, with $k_N \in \{k_{\text{min}},\dots,k_{\text{max}}\}$ and $\vec{p}_N \in \{\vec{e}_1,\vec{e}_2\}$, exists such that
  \begin{align}
	  \abs{\vec{p}_N \cdot \vec{\chi}_{s}(2T_I k_N)}^2 > 0 \quad \forall s\in \mathcal{S}.
  \end{align}
  For any given $N$, the tuple $(k_N, \vec{p}_N)$ can be found in polynomial time.
	\label{thm:ExponentialSuppression}
\end{theorem}
\begin{proof}[Proof of Theorem~\ref{thm:ExponentialSuppression}]
	As already discussed in the main letter, the condition~\eqref{eqn:Supp:ConditionExactComplexityContinuous} for a fixed $N$ directly implies the existence of at least one time bin $k_N$ and one polarization $\vec{p}_N$ for which all $N$ amplitudes $\vec{p}_N \cdot \vec{\chi}_s(2T_I k_N)$, with $s\in \mathcal{S}$, are non-vanishing.

	Moreover, due to the finite number $L\defeq k_{\text{max}} - k_{\text{min}} + 1$ of time bins in the interval $[k_{\text{min}},k_{\text{max}}]$, $(k_N,\vec{p}_N)$ can be found in polynomial time.
\end{proof}

  We can now formulate the main theorem on the complexity of exact MBCS:
  \begin{theorem}[Main theorem on the complexity of exact MBCS]
	  If $\mathcal{O}_E$ is an exact MBCS oracle for a given family $E\in \mathcal{E}$, then $\mathrm{P}^{\mathrm{\#P}} \subseteq \mathrm{BPP^{NP^{\mathcal{O}_{\mathit{E}}}}}$. This implies that the polynomial hierarchy collapses to the third level if the exact MBCS problem for states in the family $E$ can be solved in polynomial time by a classical computer.
	\label{thm:ComplexityOfExactMBCS}
\end{theorem}

\begin{proof}[Proof of Theorem~\ref{thm:ComplexityOfExactMBCS}]
	If a matrix $X \in \mathds{R}^{N\times N}$ and a parameter $g\in \left[ 1+1/\text{poly}(N),\text{poly}(N) \right]$ are given, it is \#P-hard to approximate $\abs{\per X}^2$ to within a multiplicative factor $g$ \cite{S_aaronson2011computational}.
	
	We now show that, given an MBCS oracle $\mathcal{O}_E$ for a family $E\in \mathcal{E}$ (with $\mathcal{E}$ defined in Definition~\ref{def:Supp:ComplexFamiliesExact}), it is possible to perform this approximation in FBPP\textsuperscript{NP\textsuperscript{$\mathcal{O}_E$}}.

	As shown in \cite{S_aaronson2011computational}, given $\gamma\defeq 1/\|X\| \geq 2^{-\text{poly}(N)}$, it is possible for any $M\geq 2N$ to find in polynomial time a unitary $M \times M$ matrix $\mathcal{U}$ which contains $\gamma X$ as its top-left $N\times N$ submatrix, i.e.
\begin{align}
	\gamma X = [\mathcal{U}_{d,s}]_{\substack{d=1,\dots,N\\s=1,\dots,N}}.
\end{align}

Given $\mathcal{S}\defeq\{1,2,\dots,N\}$, the matrix $A= [\mathcal{U}_{d,s}]_{\substack{d=1,\dots,M\\s\in\mathcal{S}}}\in \mathfrak{U}_{M,N}$ and the spectra $\{\vec{\xi}_s\}_{s\in \mathcal{S}}\in E_N$ induce the probability distribution $\mathcal{P}_E = \mathcal{P}_E(\{\vec{\xi}_s\},A)$. From Theorem~\ref{thm:ExponentialSuppression}, a tuple $(k_N,\vec{p}_N)$ exists such that $\prod_{s\in\mathcal{S}}\abs{\vec{p}_N \cdot \vec{\chi}_s(2T_I k_N)}^2 > 0$. For the sample
	\begin{align}
		\mathfrak{S}^* &\defeq (\mathcal{D}^*; \{\underbrace{k_N,\dots,k_N}_{N \text{ times}}\},\{\underbrace{\vec{p}_N,\dots,\vec{p}_N}_{N\text{ times}}\}),
		\intertext{with}
		\mathcal{D}^* &\defeq \{1,2,\dots,N\},
	\end{align}
	the probability in Eq.~\eqref{eqn:Supp:ProbabilityIdenticaTimePol} becomes
	\begin{align}
		p_{\mathfrak{S}^*} = \Pr_{\mathcal{P}_E}[\mathfrak{S}^*] = \gamma^{2N} (2T_I)^N \prod_{s\in\mathcal{S}} \abs{\vec{p}_N\cdot \vec{\chi}_s(t_N)}^2 \abs{\per X}^2.
		\label{eqn:Supp:ProbabilitySpecialSample}
	\end{align}

	According to Definition~\ref{def:Supp:Oracle}, the exact MBCS oracle $\mathcal{O}_E$ samples from the probability distribution $\mathcal{P}_{\mathcal{O}_E}(\{\vec{\xi}_s\},A;\beta)\equiv\mathcal{P}_E(\{\vec{\xi}_s\},A)$ returning the value $\mathcal{O}_E(\{\vec{\xi}_s\},A,r) = \mathfrak{S}^*$ with probability
	\begin{align}
		p_{\mathfrak{S}^*} = \Pr_{r\in \{0,1\}^{\text{poly}(N)}}[\mathcal{O}_E(\{\vec{\xi}_s\},A,r) = \mathfrak{S}^*].
	\end{align}
	Defining the Boolean function
	\begin{align}
		f: \{0,1\}^{\text{poly}(N)}\rightarrow \{0,1\}, r \mapsto 
		\begin{dcases}
			1 & \mathcal{O}_{E}(\{\vec{\xi}_s\},A,r) = \mathfrak{S}^* \\
			0 & \mathcal{O}_{E}(\{\vec{\xi}_s\},A,r) \neq \mathfrak{S}^*
		\end{dcases},
		\label{eqn:Supp:BooleanFunction}
	\end{align}
	we can also express this probability as
	\begin{align}
		p_{\mathfrak{S}^*} = \frac{1}{2^{\text{poly}(N)}} \sum_{r\in \{0,1\}^{\text{poly}(N)}} f(r).
		\label{eqn:Supp:ProbabilityAsSum}
	\end{align}

	Therefore, we can use Stockmeyer's algorithm \cite{S_Stockmeyer1983} to approximate $p_{\mathfrak{S}^*}$ to within a multiplicative factor $g\in [1+1/\text{poly}(N), \text{poly}(N)]$ in FBPP\textsuperscript{NP\textsuperscript{$\mathcal{O}_E$}} in time polynomial in $N$. Consequently, from Eq.~\eqref{eqn:Supp:ProbabilitySpecialSample}, we can approximate $\abs{\per X}^2$ in $\mathrm{FBPP^{NP^{\mathcal{O}_{\mathit{E}}}}}$ as well. Since performing such an approximation is a \#P-hard problem \cite{S_aaronson2011computational}, $\mathrm{P}^{\#\mathrm{P}} \subseteq \mathrm{BPP}^{\mathrm{NP}^{\mathcal{O}_E}}$. If the MBCS problem could be solved in polynomial time by a classical computer this would imply $\mathrm{P^{\#P}}\subseteq \mathrm{BPP^{NP}}$, which, by Toda's theorem \cite{S_Toda1991}, would lead to a collapse of the polynomial hierarchy to the third level.
\end{proof}
	
	\section{The approximate MBCS problem}\label{sec:Supp:ApproximateCase}
	We consider, for simplicity, a multiboson correlation experiment with polarization-insensitive detectors and states from a family $R:=\{R_N\}_{N\in \mathds{N}}\in \mathcal{F}$ (with $\mathcal{F}$ defined in Definition~\ref{def:Families}) of input states $\{\ket{1[\vec{\xi}_s]}_s\}_{s\in \mathcal{S}}$ defined by input spectra
	\begin{align}
		\{\vec{\xi}_s(\omega)\}_{s\in \mathcal{S}} = \left\{ \vec{v} \frac{1}{\sqrt{\pi \Delta\omega}} \operatorname{sinc} \left( \frac{\omega-\omega_s}{\Delta \omega} \right) \ee{\ii \omega t_{0}} \right\}_{s\in\mathcal{S}}
		\label{}
	\end{align}
	with equal emission times $t_{0s} = t_0$, equal polarizations $\vec{v}_s = \vec{v}$, and equal bandwidths $\Delta\omega_{s} = \Delta\omega$ but different colors $\omega_s$. For these spectra, the elements of the $N$-photon interference matrix $a(s,s') \defeq \abs{\vec{v}_s \cdot \vec{v}_{s'}} \int_{-\infty}^{\infty} \d{t} \abs{\chi_{s}(t - t_{0s})} \abs{\chi_{s'}(t-t_{0s'})}$ fulfill the condition
	\begin{align}
		a(s,s')  \simeq 1 \ \forall s,s'
		\label{}
	\end{align}
	of full temporal overlap.
	As in section~\ref{sec:Supp:ExactCase}, we discretize the detection-time axis into time bins of width $2T_I$. Here, as described in the main letter, in each output port a detection can only occur in the time bins $t_k \defeq 2T_I k \in T \defeq [t_0+\Delta t - 1/\Delta\omega,t_0+\Delta t + 1/\Delta \omega]$, with $k\in \{k_{\text{min}},k_{\text{min}}+1,\dots,k_{\text{max}}\}$, where $k_{\text{min}} \defeq (t_0+\Delta t - 1/\Delta\omega + T_I)/(2T_I)$ and $k_{\text{max}} \defeq (t_0 + \Delta t + 1/\Delta\omega - T_I)/(2T_I)$. Therefore, the total number of time bins in which a detection is possible is $L\defeq1/(T_I\Delta\omega)$.
	Given the output probability distribution
	\begin{align}
		\mathcal{P}_R \defeq \mathcal{P}_R(\{\vec{\xi}_s\}_{s\in \mathcal{S}}\in R_N, A \in \mathfrak{U}_{M,N}),
		\label{}
	\end{align}
	for a given sample $\mathfrak{S} \defeq (\mathcal{D};\{k_d\})$, the probability
	\begin{align}
		p_{\mathfrak{S}} \defeq \Pr_{\mathcal{P}_R}[\mathfrak{S}] = 
	  P^{(\mathcal{D},\mathcal{S})}_{\{t_{k_d}\}} =
	  (T_I\Delta\omega)^N \abs{\per \Big( [\mathcal{U}^{(\mathcal{D},\mathcal{S})}_{d,s} \ee{\ii \omega_s t_{k_d}}]_{\substack{d\in\mathcal{D}\\s\in\mathcal{S}}}\Big)}^2,
		\label{}
	\end{align}
	derived in the main letter, takes the form
  \begin{align}
	  p_{\mathfrak{S}} = 
	  \frac{1}{L^N} \abs{\per \Big( [\mathcal{U}^{(\mathcal{D},\mathcal{S})}_{d,s} \ee{2 \ii T_I \omega_s k_d}]_{\substack{d\in\mathcal{D}\\s\in\mathcal{S}}}\Big)}^2.
	  \label{eqn:Supp:Probabilities}
  \end{align}

  We can now formulate the following theorem on the computational power of an approximate MBCS oracle for the family $R$:
  \begin{theorem}[Computational power of an approximate MBCS oracle for the family $R$ of input states]
	  Let $\mathcal{O}_R$ be an approximate oracle for the MBCS problem for the family $R$ and $\varepsilon,\delta>0$ given error bounds.
	  It is then possible to approximate the modulus square of the permanent of a random Gaussian $N\times N$ matrix to within an additive error $\pm \varepsilon N!$ and with success probability larger than $1-\delta$ in $\mathrm{FBPP^{NP^{\mathcal{O}_{\mathit{R}}}}}$ in time polynomial in $N$, $1/\varepsilon$, and $1/\delta$.
	  \label{thm:PowerApproximateOracle}
  \end{theorem}

  \begin{proof}[Proof of Theorem~\ref{thm:PowerApproximateOracle}]
	  Let $X$ be a complex $N \times N$ matrix whose entries are randomly picked according to a standard complex normal probability distribution $\mathcal{N}(0,1)_{\mathds{C}}$.
	  We will adapt the arguments found in~\cite{S_aaronson2011computational} to give an algorithm in $\mathrm{FBPP}^{\mathrm{NP}^{\mathcal{O}_R}}$ that performs the required approximation in polynomial time. 
	  The main idea is to introduce an $M\times N$ matrix $A=[\mathcal{U}_{d,s}]_{\substack{d=1,\dots,M\\s\in \mathcal{S}}}\in \mathfrak{U}_{M,N}$ (as in the exact case, $A$ is a rectangular submatrix of a Haar-random unitary interferometer matrix $\mathcal{U}$ corresponding to the input configuration $\mathcal{S} \defeq \{1,2,\dots,N\}$) such that a specific sample $\mathfrak{S}^{*}\defeq(\mathcal{D}^*,\{k_d^*\})$ occurs with probability $p_{\mathfrak{S}^*} \propto \abs{\per X}^2$, which can be estimated using the approximate MBCS oracle defined in Definition~\ref{def:Supp:Oracle}.

	  \begin{enumerate}
		  \item In order to rule out the possibility that the oracle can willingly sabotage the output probability of the sample $\mathfrak{S}^{*}$ in which we are interested, $\mathfrak{S}^*$ has to be picked randomly. Therefore, the time sample $\{k_d^*\}$ is generated by randomly picking $N$ time bins $t_{k^*_d} = 2T_I k_d^* \in T$ according to a uniform probability distribution. 
  Given the structure of the probabilities in Eq.~\eqref{eqn:Supp:Probabilities}, it is useful to define a matrix $\tilde{X}$ which incorporates the inverse of the phase factors corresponding to the randomly chosen time sample $\{k_d^*\}$, i.e.
  \begin{align}
	  \tilde{X}_{i,j}\defeq X_{i,j}\ee{-2 \ii T_I \omega_{s_j} k^*_{d_i}},
  \end{align}
  where $d_i$ and $s_i$ are the $i$-th elements of $\mathcal{D}$ and $\mathcal{S}$, respectively.
  Since the elements of $X$ are i.i.d. Gaussian random variables, the same holds for the elements of $\tilde{X}$, as shown in App.~\ref{app:Supp:IIDGaussians}. 
	 
  \item An $M\times N$ matrix $A\in \mathfrak{U}_{M,N}$ is generated with the following properties: $A$ is distributed like the submatrix $[\mathcal{U}_{d,s}]_{\substack{d=1,\dots,M\\s\in \mathcal{S}}}$ of a Haar-distributed $M\times M$ unitary matrix $\mathcal{U}$ and it contains $\tilde{X}' \defeq \tilde{X}/\sqrt{M}$ as a uniformly random submatrix. 
	  As shown in \cite{S_aaronson2011computational}, since $\tilde{X}$ is a Gaussian matrix this is possible to achieve in BPP\textsuperscript{NP} with a failure probability 
	  \begin{align}
		  \Pr[\text{hiding failed}] \leq \frac{\delta}{4},
		  \label{eqn:Supp:ProbabilitySuccessHiding}
		  \intertext{provided}
			M = K \frac{N^5}{\delta}\log^2\frac{N}{\delta}
			\label{eqn:Supp:MinimumMApproximate}
	  \end{align}
	  with a sufficiently large constant $K$.
	  The random output-port sample that corresponds to the position of $\tilde{X}' = [\mathcal{U}_{d,s}]_{\substack{d\in \mathcal{D}^*\\ s\in\mathcal{S}}}$ inside of $A$ is denoted as $\mathcal{D}^*$.
	  Thus, for the random sample $\mathfrak{S}^*\defeq(\mathcal{D}^*;\{k_d^*\})$, Eq.~\eqref{eqn:Supp:Probabilities} becomes
  \begin{align}
	  p_{\mathfrak{S}^*} = \frac{1}{L^N} \abs{\per\left( \big[M^{-1/2}\tilde{X}_{i,j}\ee{2\ii T_I \omega_{s_j}k^*_{d_i}}\big]_{\substack{i=1,\dots,N\\j=1,\dots,N}} \right)}^2  = \frac{1}{L^N M^N} \abs{\per X}^2.
	  \label{eqn:Supp:DetectionProbabilitySpecialSample}
  \end{align}
  \item  While $p_{\mathfrak{S}^*}$ describes the probability for the sample $\mathfrak{S}^{*}$ within the exact probability distribution $\mathcal{P}_R = \mathcal{P}_R(\{\vec{\xi}_s\},A)$, we define
  \begin{align}
	  q_{\mathfrak{S}^{*}} \defeq \Pr_{\mathcal{P}_{\mathcal{O}_R}}[\mathfrak{S}^{*}] = \Pr_{r\in\{0,1\}^{\text{poly}(N)}}\left[ \mathcal{O}_R(\{\vec{\xi}_s\},A,r;\beta) = \mathfrak{S}^* \right]
  \end{align}
  as the probability for the sample $\mathfrak{S}^{*}$ within the probability distribution $\mathcal{P}_{\mathcal{O}_R} = \mathcal{P}_{\mathcal{O}_R}(\{\vec{\xi}_s\},A;\beta)$ of the approximate MBCS oracle $\mathcal{O}_R$ (as shown in App.~\ref{app:Supp:Binary}, all probabilities $p_{\mathfrak{S}}$ are at most exponentially small and it is thus ensured that an oracle $\mathcal{O}_R$ as defined in Definition~\ref{def:Supp:Oracle} can sample according to a probability distribution $\mathcal{P}_{\mathcal{O}_R}(\{\vec{\xi}_s\},A;\beta)$ arbitrarily close to $\mathcal{P}_R(\{\vec{\xi}_s\},A)$).

  In analogy to Eqs.~\eqref{eqn:Supp:BooleanFunction} and \eqref{eqn:Supp:ProbabilityAsSum}, we can also define a Boolean function
  \begin{align}
	  f: \{0,1\}^{\text{poly}(N)} \rightarrow \{0,1\}, r \mapsto  
	  \begin{dcases}
		  1 & \mathcal{O}_R(\{\vec{\xi}_s\},A,r;\beta) = \mathfrak{S}^* \\
		  0 & \mathcal{O}_R(\{\vec{\xi}_s\},A,r;\beta) \neq \mathfrak{S}^*
	  \end{dcases},
  \end{align}
  and write
  \begin{align}
	  q_{\mathfrak{S}^*} = \frac{1}{2^{\text{poly}(N)}} \sum_{r\in \{0,1\}^{\text{poly}(N)}} f(r).
  \end{align}
  This makes it evident that Stockmeyer's algorithm \cite{S_Stockmeyer1983} can be used to find a value $\tilde{q}_{\mathfrak{S}^*}$ in $\mathrm{FBPP^{NP^{\mathcal{O}_{\mathit{R}}}}}$ that approximates $q_{\mathfrak{S}^*}$ to within a multiplicative factor $(1+\alpha) \in [1+ 1/\text{poly}(N), \text{poly}(N)]$, such that
  \begin{align}
	  \Pr \left[ q_{\mathfrak{S}^*}/(1+\alpha) \leq \tilde{q}_{\mathfrak{S}^*} \leq (1+\alpha) q_{\mathfrak{S}^*} \right] \geq 1 - \frac{1}{2^M}
	  \label{eqn:Supp:StockmeyerSuccessProbability}
  \end{align}
  in time polynomial in $M$ and $1/\alpha$ \cite{S_Stockmeyer1983}.
	\item By using Definition~\ref{def:Supp:Oracle} of the approximate MBCS oracle and Eq.~\eqref{eqn:Supp:StockmeyerSuccessProbability} for $\alpha =\varepsilon \delta/16$, we show in App.~\ref{app:Supp:FailureBound} that the estimate $L^N M^N \tilde{q}_{\mathfrak{S}^*}$ of the squared permanent of $X$
		satisfies the condition
		\begin{align}
			\Pr\left[ \abs{\abs{\per X}^2 - L^N M^N \tilde{q}_{\mathfrak{S}^*}} > \varepsilon N! \right]< \frac{\delta}{2} + \frac{1}{2^M}.
			\label{eqn:Supp:MainInequality}
		\end{align}
  \end{enumerate}
  Adding to Eq.~\eqref{eqn:Supp:MainInequality} the probability of failure for the hiding of $\tilde{X}'$ from Eq.~\eqref{eqn:Supp:ProbabilitySuccessHiding} and recalling Eq.~\eqref{eqn:Supp:MinimumMApproximate} which ensures $2^{-M} < \delta/4$, we find that the total probability of failure is smaller than $\delta$, as required.

		Further, the algorithm runs in a time polynomial in $N$, $1/\delta$, and $1/\varepsilon$ since the hiding procedure in step 2 is polynomial in time with respect to $N$ and $1/\delta$ and the running time of Stockmeyer's algorithm (used in step 3) is polynomial in $N$ and $1/\alpha$, where $\alpha=\varepsilon\delta/16$ was chosen.
  \end{proof}

  In \cite{S_aaronson2011computational}, the authors argue that approximating the modulus square of the permanent of a Gaussian matrix is a \#P-hard problem if two reasonable conjectures are true. Under this assumption, the following theorem holds:

  \begin{theorem}
	  Let $\mathcal{O}_R$ be an approximate MBCS oracle for the family $R$. If the approximate MBCS problem for $R$ can be solved in polynomial time by a classical computer, the polynomial hierarchy collapses to the third level.
	  \label{thm:ComplexityOfApproximateCase}
  \end{theorem}

  \begin{proof}[Proof of Theorem~\ref{thm:ComplexityOfApproximateCase}]
	  Theorem~\ref{thm:PowerApproximateOracle} states that it is possible to approximate the modulus square of a permanent in $\mathrm{FBPP^{NP^{\mathcal{O}_{\mathit{R}}}}}$, given an oracle $\mathcal{O}_R$ for the family $R$. If this approximation is indeed \#P-hard, as argued in \cite{S_aaronson2011computational}, it follows that $\mathrm{P^{\#P}} \subseteq \mathrm{BPP^{NP^{\mathcal{O}_{R}}}}$. Further, if the MBCS problem for states of this family can be solved in polynomial time by a classical computer, then $\mathrm{P^{\#P}} \subseteq \mathrm{BPP^{NP}}$ which implies, by Toda's theorem \cite{S_Toda1991}, that the polynomial hierarchy collapses to the third level.
  \end{proof}
  
  \appendix
\renewcommand{\theequation}{S\Alph{section}\arabic{equation}}
  \section{Exponential lower bound on the MBCS probabilities}\label{app:Supp:Binary}
  
  Here, we show that the assumption of a rational representation $(x+\ii y)/2^{\text{poly}(N)}$ (with integers $x,y$) of the elements of the interferometer matrix and of the values of the temporal distributions for all possible time bins results in an exponential lower bound on the non-vanishing detection probabilities $p_{\mathfrak{S}}$.

  Using the expression~\eqref{eqn:Supp:ProbabilitiesExact}, the definitions of a matrix permanent and of the matrix $\mathcal{T}^{(\mathcal{D},\mathcal{S})}_{\{k_d,\vec{p}_d\}}$ in Eq.~\eqref{eqn:Supp:Tmatrix}, and the explicit expression Eq.~\eqref{eqn:Supp:TemporalSpectra} for the functions $\vec{\chi}_s(t)$, we find
  \begin{align}
	  p_{\mathfrak{S}} &= (2T_I)^N\sum_{\sigma,\sigma'} \bigg[ \phantom{\times} \prod_{d\in \mathcal{D}} \mathcal{U}_{d,\sigma'(d)}^{*} \mathcal{U}_{d,\sigma(d)} (\vec{p}_d \cdot \vec{v}_{\sigma'(d)})^*  (\vec{p}_d \cdot \vec{v}_{\sigma(d)})\\ 
		  &\phantom{=\Sigma \Big[\times}\times \prod_{d\in\mathcal{D}} (\chi_{\sigma'(d)}(t_d-t_{0\sigma'(d)}-\Delta t))^* \chi_{\sigma(d)}(t_d - t_{0\sigma(d)}-\Delta t) \\
	  &\phantom{=\Sigma \Big[\times}\times \prod_{d\in\mathcal{D}} \ee{-\ii \omega_{\sigma'(d)}(t_d - t_{0\sigma'(d)}-\Delta t)} \ee{\ii \omega_{\sigma(d)}(t_d - t_{0\sigma(d)} - \Delta t)} \bigg] 
	  \label{}
  \end{align}
  Since all $8N+1$ numbers $T_I$, $\mathcal{U}_{d,s}$ and c.c., $(\vec{p}_d\cdot \vec{v}_s)$ and c.c., $\chi_s(t_d-t_{0s}-\Delta t)$ and c.c., $\exp[\ii \omega_s (t_d-t_{0s}-\Delta t)]$ and c.c. are represented as $(x+\ii y)/2^{\text{poly}(N)}$ ($x$,$y$ integers), it is immediately clear that the probability has the form
  \begin{align}
	  p_{\mathfrak{S}}  = \frac{1}{2^{\text{poly}(N)}} w_{\mathfrak{S}},
	  \label{eqn:Supp:ProbabilitiesBinary}
  \end{align}
  with a non-negative integer $w_{\mathfrak{S}}$ and is therefore at most exponentially small. This guarantees that Stockmeyer's algorithm can be used to approximate these probabilities.

  The probabilities are also at most exponentially small if these values are represented by algebraic numbers. Indeed, it was shown in \cite{S_Kuperberg2015} that
  \begin{align}
	  p_{\mathfrak{S}} \geq 2^{-r(N)},
	  \label{eqn:Supp:AppB:LowerBoundProbs}
  \end{align}
 with a polynomial $r(N)$.

  \section{Using algebraic numbers instead of rational numbers with polynomial precision}\label{app:Supp:Algebraic}
  Here, we will show that, instead of assuming that the values $\mathcal{U}_{d,s}$, $T_I$, $\vec{\chi}_s(t_d-t_{0s}-\Delta t)$, and $\exp[\ii \omega_{s} (t_d-t_{0s}-\Delta t)]$ are represented as rational numbers $(x+\ii y)/2^{\text{poly}(N)}$ as in \cite{S_aaronson2011computational}, the hardness of MBCS sampling can also be proven if these values are represented by algebraic numbers \footnote{This possibility was brought to our attention by S. Aaronson in a private communication.}. This is appealing because it is possible to find an arbitrarily close algebraic approximation of any complex number since the algebraic numbers are dense in $\mathds{C}$. 

  We emphasize that, in general, an MBCS oracle as defined in Definition~\ref{def:Supp:Oracle} now cannot sample from the exact probability distribution $\mathcal{P}_F$ in Eq.~\eqref{eqn:Supp:Distribution} any more. However, as we will demonstrate in the following, the hardness proof in section~\ref{sec:Supp:ExactCase} is still valid if we define an ``exact'' MBCS oracle as an oracle which samples from a probability distribution $\mathcal{P}_{O_F}$ with
  \begin{align}
	  \| \mathcal{P}_{\mathcal{O}_F}(\{\vec{\xi}_s\},A) - \mathcal{P}_F(\{\vec{\xi}_s\},A) \| \leq 2^{-s(N)},
	  \label{eqn:Supp:AppB:UpperBoundDist}
  \end{align}
  where the polynomial $s(N)$ is assumed to dominate the polynomial $r(N)$ from Eq.~\eqref{eqn:Supp:AppB:LowerBoundProbs} ($s(N)-r(N)$ grows monotonically, and $s(1)>r(1)+2$). Such an approximation of the exact probability distribution can however still be achieved by an oracle as defined in Definition~\ref{def:Supp:Oracle}.

  Defining the probabilities $q_{\mathfrak{S}} \defeq \Pr_{\mathcal{P}_{O_F}}[\mathfrak{S}]$, Eq.~\eqref{eqn:Supp:AppB:UpperBoundDist} implies that
  \begin{align}
	  \abs{p_{\mathfrak{S}} - q_{\mathfrak{S}}} &\leq 2 \cdot 2^{-s(N)} \\
	  \Leftrightarrow p_{\mathfrak{S}} - 2\cdot 2^{-s(N)} &\leq q_{\mathfrak{S}} \leq p_{\mathfrak{S}} + 2\cdot 2^{-s(N)} \\
	  \Rightarrow p_{\mathfrak{S}}\left(1 - 2\frac{2^{-s(N)}}{p_{\mathfrak{S}}}\right) &\leq q_{\mathfrak{S}} \leq p_{\mathfrak{S}}\left(1 + 2\frac{2^{-s(N)}}{p_{\mathfrak{S}}}\right),
	  \intertext{which with Ineq.~\eqref{eqn:Supp:AppB:LowerBoundProbs} becomes}
	   p_{\mathfrak{S}}\left(1 - 2\frac{2^{-s(N)}}{2^{-r(N)}}\right) &\leq q_{\mathfrak{S}} \leq p_{\mathfrak{S}}\left(1 + 2\frac{2^{-s(N)}}{2^{-r(N)}}\right) \\
	  \Rightarrow p_{\mathfrak{S}}\left(1 - 2^{-(s(N)-r(N)-1)}\right) &\leq q_{\mathfrak{S}} \leq p_{\mathfrak{S}}\left(1 + 2^{-(s(N)-r(N)-1)}\right) .
	  \label{}
  \end{align}
  Using the inequality $1-x \geq 1/(1+2x)$ for $0<x < 1/2$, we find that
  \begin{align}
	  	 p_{\mathfrak{S}}\frac{1}{1 + 2^{-(s(N)-r(N)-2)}} &\leq q_{\mathfrak{S}} \leq p_{\mathfrak{S}}\left(1 + 2^{-(s(N)-r(N)-2)}\right),
	  \label{eqn:Supp:ExpClose}
  \end{align}
  i.e. that $q_{\mathfrak{S}}$ is multiplicatively close to $p_{\mathfrak{S}}$ to within a factor $1+2^{-\text{poly}(N)}$.

  The approximation of $p_{\mathfrak{S}}$ to within a multiplicative factor $g\in[1+\frac{1}{\text{poly}(N)},\text{poly}(N)]$ can therefore be achieved by approximating $q_{\mathfrak{S}}$ to within a factor $g'\defeq g/(1+2^{-(s(N)-r(N)-2)})$. Since $g-1$ is at most polynomially small, $g'\in[1+\frac{1}{\text{poly}(N)},\text{poly}(N)]$ as well. Therefore, Stockmeyer's algorithm tells us that a value $\tilde{q}_{\mathfrak{S}}$ can be found in $\mathrm{BPP^{NP}}$, such that
  \begin{align}
	  \frac{q_{\mathfrak{S}}}{g'} \leq \tilde{q}_{\mathfrak{S}} \leq g' q_{\mathfrak{S}}.
	  \label{}
  \end{align}
  Then, Ineq.~\eqref{eqn:Supp:ExpClose} implies that
  \begin{align}
	  \frac{p_{\mathfrak{S}}}{g} = \frac{p_{\mathfrak{S}}}{g'(1+2^{-(s(N)-r(N)-2)}} \leq \tilde{q}_{\mathfrak{S}} \leq p_{\mathfrak{S}} g' (1+2^{-(s(N)-r(N)-2)}) = p_{\mathfrak{S}} g.
	  \label{}
  \end{align}
   and the arguments in the hardness proof for the exact MBCS in section~\ref{sec:Supp:ExactCase} still hold if we define ``exact'' in the sense of Ineq.~\eqref{eqn:Supp:AppB:UpperBoundDist}.

	Further, the proof for the hardness of the approximate MBCS problem from section~\ref{sec:Supp:ApproximateCase} is still completely valid in the algebraic number representation since an approximate MBCS oracle as defined in Definition~\ref{def:Supp:Oracle} is still able to sample from a probability distribution that is polynomially close in variation distance to the exact probability distribution.

 \section{The distribution of the elements of \texorpdfstring{$\tilde{X}$}{X}}\label{app:Supp:IIDGaussians}

 We will now show that, under the assumption that the elements $X_{i,j}$, $i,j=1,\dots,N$, of a complex $N\times N$ matrix $X$ are i.i.d. random variables with a complex standard normal distribution, the same holds for the elements of $\tilde{X}$ which are defined as
  \begin{align}
	  \tilde{X}_{i,j} = X_{i,j}\ee{\ii \varphi_{i,j}}.
  \end{align}

  If the elements of $X$ are i.i.d. $\mathcal{N}(0,1)_{\mathds{C}}$ variables, their joint probability distribution is
  \begin{align}
	  f_{X}(\{X_{i,j}\}) = \prod_{i,j=1}^{N} \frac{1}{\pi} \ee{-\abs{X_{i,j}}^2}.
  \end{align}

  Noting that the Jacobi determinant for the change of variables between the $\{X_{i,j}\}$ and the $\{\tilde{X}_{i,j}\}$ is $\det \vec{J} = 1$, we find that the common probability distribution for the elements of $\tilde{X}$ is
  \begin{align}
	  f_{\tilde{X}}(\{\tilde{X}_{i,j}\}) = f_{X}(\{\tilde{X}_{i,j} \ee{-\ii \varphi_{i,j}}\}) \cdot \underbrace{\det \vec{J}}_{=1} = \prod_{i,j=1}^N \frac{1}{\pi} \ee{- \abs{\tilde{X}_{i,j}}^2},
  \end{align}
  where we used that the complex Gaussian distributions are independent of the complex phase.

  Therefore, the elements of the matrix $\tilde{X}$ are still i.i.d. $\mathcal{N}(0,1)_{\mathds{C}}$ random variables if the elements of $X$ were.

  \section{Bound on failure probability of approximation}\label{app:Supp:FailureBound}

  As in \cite{S_aaronson2011computational}, we define $\Phi_{M,N}$ as the set of all port samples $\mathcal{D}$ and $G_{M,N}$ as the set of bunching-free port samples, i.e. as the set of port samples $\mathcal{D}$ which consist only of pairwisely different output port indices. 
  Further, let $\Omega_{L,N}$ be the set of all time samples $\{k_d\}$.

  We will now proceed to derive three inequalities which combined yield Ineq.~\eqref{eqn:Supp:MainInequality}.

  \begin{enumerate}
	  \item With $\Delta_{\mathfrak{S}} \defeq \abs{p_{\mathfrak{S}}-q_{\mathfrak{S}}}$, we find the expectation value
  \begin{align}
	  \EV_{\mathfrak{S}\in G_{M,N}\otimes \Omega_{L,N}}[\Delta_{\mathfrak{S}}] \defeq \frac{\sum_{\mathfrak{S}\in G_{M,N}\otimes\Omega_{L,N}} \Delta_{\mathfrak{S}}}{\abs{G_{M,N}} \abs{\Omega_{L,N}}} \leq \frac{\sum_{\mathfrak{S}\in \Phi_{M,N}\otimes\Omega_{L,N}} \Delta_{\mathfrak{S}}}{\abs{G_{M,N}} \abs{\Omega_{L,N}}}\\
	  \stackrel{\text{def}}{=} \frac{2 \big\| \mathcal{P}_{\mathcal{O}_R} - \mathcal{P}_R \big\|}{\abs{G_{M,N}}\abs{\Omega_{L,N}}} = \frac{2 \big\| \mathcal{P}_{\mathcal{O}_R} - \mathcal{P}_R \big\|}{\binom{M}{N} L^N} \leq \frac{2\beta}{\binom{M}{N} L^N} < 3\beta \frac{N!}{(L M)^N},
  \end{align}
  where, in the penultimate step, we used the Definition~\ref{def:Supp:Oracle} of an approximate MBCS oracle and, in the last step, the fact that $M = \omega(N^2)$ ($M$ asymptotically grows faster than $N^2$) \cite{S_aaronson2011computational}.
  Therefore, Markov's inequality gives
  \begin{align}
	  \Pr_{\mathfrak{S}\in G_{M,N}\otimes\Omega_{L,N}}\left[ \Delta_{\mathfrak{S}} > 3\beta \frac{N!}{(LM)^N} \frac{4}{\delta} \right] < \frac{\delta}{4} 
  \end{align}
  which, by choosing $\beta = \varepsilon \delta/24$, simplifies to
  \begin{align}
	  \Pr_{\mathfrak{S}\in G_{M,N}\otimes\Omega_{L,N}}\left[ \Delta_{\mathfrak{S}} > \frac{\varepsilon}{2} \frac{N!}{(LM)^N} \right] < \frac{\delta}{4} .
	  \label{eqn:Supp:MarkovInequalityIGeneralSample}
  \end{align}
  The MBCS oracle only knows $A$ but has no information whatsoever about which sample $\mathfrak{S}^{*}$ has been chosen from $G_{M,N}\otimes \Omega_{L,N}$. Thus, Eq.~\eqref{eqn:Supp:MarkovInequalityIGeneralSample} implies
  \begin{align}
	  \Pr_{X, A}\left[ \Delta_{\mathfrak{S}^*} > \frac{\varepsilon}{2} \frac{N!}{(LM)^N} \right] < \frac{\delta}{4}.
	  \label{eqn:Supp:MarkovInequalityISpecialSample}
  \end{align}
\item 
	As shown in step 3 of the proof of Theorem~\ref{thm:PowerApproximateOracle}, we can apply Stockmeyer's algorithm \cite{S_Stockmeyer1983} to find an approximate $\tilde{q}_{\mathfrak{S}^*}$ of $q_{\mathfrak{S}^*}$ in $\mathrm{FBPP^{NP^{\mathcal{O}_{\mathit{R}}}}}$. 
	The algorithm guarantees that, for an arbitrary $\alpha>0$ and a runtime polynomial in $1/\alpha$ and $M$,
  \begin{align}
	  \Pr[\abs{\tilde{q}_{\mathfrak{S}^*} - q_{\mathfrak{S}^*}}> \alpha\, q_{\mathfrak{S}^*}] \leq \Pr[\tilde{q}_{\mathfrak{S}^*} > (1+\alpha)q_{\mathfrak{S}^*} \lor \tilde{q}_{\mathfrak{S}^*} < q_{\mathfrak{S}^*}/(1+\alpha)] < \frac{1}{2^M},
	  \label{eqn:Supp:InequalityStockmeyerII}
  \end{align}
  where the first inequality follows from $1-\alpha < 1/(1+\alpha)$ and the second inequality is equivalent to Ineq.~\eqref{eqn:Supp:StockmeyerSuccessProbability}. With the choice $\alpha=\varepsilon\delta/16$, this becomes
  \begin{align}
	  \Pr[\abs{\tilde{q}_{\mathfrak{S}^*} - q_{\mathfrak{S}^*}}> \frac{\varepsilon\delta}{16}\, q_{\mathfrak{S}^*}] < \frac{1}{2^M}.
	  \label{eqn:Supp:StockmeyerInequalitySpecialSample}
  \end{align}
\item Lastly,
  \begin{align}
	  \EV_{\mathfrak{S}\in G_{M,N}\otimes \Omega_{L,N}}[q_{\mathfrak{S}}] = \frac{\sum_{\mathfrak{S}\in G_{M,N}\otimes \Omega_{L,N}}q_{\mathfrak{S}}}{\abs{G_{M,N}} \abs{\Omega_{L,N}}} \leq \frac{1}{\abs{G_{M,N}}\abs{\Omega_{L,N}}} = \frac{1}{\binom{M}{N}L^N } < 2 \frac{N!}{(L M)^N}
  \end{align}
  and thus by invoking Markov's inequality
  \begin{align}
	  \Pr_{\mathfrak{S}\in G_{M,N}\otimes \Omega_{L,N}} \left[ q_{\mathfrak{S}} > 2 \frac{N!}{(LM)^2} \frac{4}{\delta} \right] < \frac{\delta}{4}.
  \end{align}
  With the same arguments leading to Ineq.~\eqref{eqn:Supp:MarkovInequalityISpecialSample}, this implies
  \begin{align}
	  \Pr_{X,A} \left[ q_{\mathfrak{S}^*} > 2 \frac{N!}{(LM)^2} \frac{4}{\delta} \right] < \frac{\delta}{4}.
	  \label{eqn:Supp:MarkovInequalityIISpecialSample}
  \end{align}

  \end{enumerate}
  Indeed, by using the inequalities
  \begin{align}
	  \Pr[\abs{a-b} > c ] &\leq \Pr\left[\abs{a} > \frac{c}{2}\right] + \Pr\left[\abs{b} > \frac{c}{2}\right] \qquad (c>0)
	  \label{eqn:Supp:InequalityProbabilityCombination1}
	  \intertext{and}
	  \Pr[\abs{a-b} > c] &\leq \Pr\left[b > \frac{c}{k}\right] + \Pr\left[ \abs{a-b} > k b \right] \qquad (c,k>0),
	  \label{eqn:Supp:InequalityProbabilityCombination2}
  \end{align}
  demonstrated in App.~\ref{app:Supp:InequalityProbabilities}, we obtain
  \begin{align}
	  \Pr\left[ \abs{\tilde{q}_{\mathfrak{S}^*} - p_{\mathfrak{S}^*}} > \varepsilon \frac{N!}{(LM)^N} \right] &\leq \Pr\left[ \abs{\tilde{q}_{\mathfrak{S}^*} - q_{\mathfrak{S}^*}} > \frac{\varepsilon}{2}\frac{N!}{(LM)^N} \right] + \Pr\left[ \abs{q_{\mathfrak{S}^*} - p_{\mathfrak{S}^*}}> \frac{\varepsilon}{2} \frac{N!}{(LM)^N} \right] \\
	  &\leq \Pr_{X,A}\left[ q_{\mathfrak{S}^*} > \frac{8}{\delta} \frac{N!}{(LM)^N} \right] + \Pr\left[\abs{\tilde{q}_{\mathfrak{S}^*} - q_{\mathfrak{S}^*}} > \frac{\varepsilon\delta}{16} q_{\mathfrak{S}^*}\right] + \Pr_{X,A}\left[ \Delta_{\mathfrak{S}^*} > \frac{\varepsilon}{2} \frac{N!}{(LM)^N} \right] \\
	  &< \frac{\delta}{4} + \frac{1}{2^M} + \frac{\delta}{4} = \frac{\delta}{2} + \frac{1}{2^M},
  \end{align}
  where in the last step, we inserted Ineqs.~\eqref{eqn:Supp:MarkovInequalityISpecialSample}, \eqref{eqn:Supp:StockmeyerInequalitySpecialSample}, and \eqref{eqn:Supp:MarkovInequalityIISpecialSample}.
  This inequality is equivalent to Ineq.~\eqref{eqn:Supp:MainInequality}, as can be seen by inserting the expression for $p_{\mathfrak{S}^*}$ from Eq.~\eqref{eqn:Supp:DetectionProbabilitySpecialSample}.

  \section{Inequalities for probabilities}\label{app:Supp:InequalityProbabilities}
  Here, we want to prove the inequalities~\eqref{eqn:Supp:InequalityProbabilityCombination1} and \eqref{eqn:Supp:InequalityProbabilityCombination2}.
  \begin{proof}
	  First, observe that
	  \begin{align}
		  \abs{a} \leq \frac{c}{2} \land \abs{b} \leq \frac{c}{2} \Rightarrow \abs{a-b} \leq c
		  \intertext{and therefore}
		  \Pr\left[ \abs{a} \leq \frac{c}{2} \land \abs{b} \leq \frac{c}{2} \right] \leq \Pr\left[ \abs{a-b} \leq c \right].
		  \label{eqn:Supp:InequalityProbabilitiesAux1}
	  \end{align}
	  Second, for two events $A$ and $B$, we know that
	  \begin{align}
		  \Pr[A \cap B] = \Pr[A] + \Pr[B] - \Pr[A \cup B]
	  \end{align}
	  and it follows that ($\bar{A},\bar{B}$ denote the complement of $A,B$)
	  \begin{align}
		  1 - \Pr[A \cap B] = 1 - \Pr[A] - \Pr[B] + \Pr[A \cup B] \leq 2 - \Pr[A] - \Pr[B] = \Pr[\bar{A}] + \Pr[\bar{B}].
		  \label{eqn:Supp:InequalityProbabilitiesAux2}
	  \end{align}
	  Now, Ineq.~\eqref{eqn:Supp:InequalityProbabilityCombination1} follows as
	  \begin{align}
		  \Pr[\abs{a-b}>c] = 1 - \Pr[\abs{a-b} \leq c] \leq 1 - \Pr\left[ \abs{a} \leq \frac{c}{2} \land \abs{b}\leq \frac{c}{2} \right] \leq \Pr\left[ \abs{a} > \frac{c}{2} \right] + \Pr\left[ \abs{b} > \frac{c}{2} \right],
	  \end{align}
	  where we used Ineq.~\eqref{eqn:Supp:InequalityProbabilitiesAux1} and Ineq.~\eqref{eqn:Supp:InequalityProbabilitiesAux2} in the second and third step, respectively.

	  Ineq.~\eqref{eqn:Supp:InequalityProbabilityCombination2} can be proved in a similar way.
  \end{proof}

\end{document}